\documentclass{llncs}

\usepackage{float}
\restylefloat{figure}

\usepackage{subfig}
\usepackage{verbatim} 
\usepackage{amssymb}
\usepackage{lipsum}
\usepackage{multicol}
\usepackage{amsmath}
\usepackage{color}
\usepackage{framed} 
\usepackage{hyperref}
\usepackage{macros}
\usepackage{url}
\usepackage{xspace}
\usepackage{xifthen}
\usepackage{listings}
\usepackage{cite}
\usepackage{paralist}
\usepackage{comment}
\usepackage{parcolumns}
\usepackage{fancyvrb}
\usepackage[draft]{fixme}
\usepackage[inline]{enumitem}
\setlist{nolistsep}
\newif\ifcode
\codefalse
\codetrue
\ifcode
\usepackage[ruled]{algorithm}
\usepackage{algpseudocode}
\usepackage{ioa_code}

\usepackage{tikz}

\usepackage{cleveref}
\crefname{section}{\S}{\S\S}
\newcommand{\RNum}[1]{\uppercase\expandafter{\romannumeral #1\relax}}

\newcounter{linenumber}

\newcommand{\true}{\mathit{true}}
\newcommand{\false}{\mathit{false}}

\newcommand{\remove}[1]{}

\newcommand{\ignore}[1]{}

\definecolor{dkgreen}{rgb}{0,0.6,0}
\definecolor{gray}{rgb}{0.5,0.5,0.5}
\definecolor{mauve}{rgb}{0.58,0,0.82}

\VerbatimFootnotes

\definecolor{gpcolor}{rgb}{0.6,0.2,0.3}
\newboolean{showcomments}
\setboolean{showcomments}{true}
\ifthenelse{\boolean{showcomments}}
{ \newcommand{\mynote}[3]{
    \fbox{\bfseries\sffamily\scriptsize#1}
    {\small$\blacktriangleright$\textsf{\emph{\color{#3}{#2}}}$\blacktriangleleft$}}
}
{
\newcommand{\mynote}[3]{}}

\newcommand{\va}[1]{\mynote{Vitaly}{#1}{magenta}}

\newcommand{\myparagraph}[1]{ \vspace{1mm}\noindent\textbf{{#1}.}}

\begin{document}

\hyphenation{con-cur-ren-cy}
\hyphenation{con-cur-ren-cy-op-ti-mal}

\bibliographystyle{splncs}

\title{A Concurrency-Optimal Binary Search Tree}
\titlerunning{A Concurrency-Optimal Binary Search Tree}
\toctitle{A Concurrency-Optimal Binary Search Tree}
\author{Vitaly Aksenov\inst{1} \and Vincent Gramoli\inst{2} \and Petr Kuznetsov\inst{3} \and Anna Malova\inst{4} \and Srivatsan Ravi\inst{5}}
\authorrunning{V. Aksenov \and V. Gramoli \and P. Kuznetsov \and A. Malova \and S. Ravi}

\institute{
  INRIA Paris / ITMO University
\and 
  University of Sydney
\and 
  LTCI, T\'el\'ecom ParisTech Universit\'e Paris-Saclay
\and
  Washington University in St Louis
\and
  Purdue University
\date{}}


\maketitle

\begin{abstract}
The paper presents the first \emph{concurrency-optimal} implementation of a
binary search tree (BST). The implementation, based on a standard sequential implementation of
a partially-external tree, ensures that every \emph{schedule}, i.e.,
interleaving of steps of the sequential code, is accepted unless linearizability
is violated. To ensure this property, we use a novel read-write locking protocol
that protects tree \emph{edges} in addition to its nodes.  

Our implementation performs comparably to the state-of-the-art BSTs 
and even outperforms them on few workloads, which suggests that 
optimizing the set of accepted schedules of the sequential code can be an
adequate design principle for efficient concurrent data structures. 

\vspace{-3em}
\end{abstract}


\section{Introduction}
\label{sec:intro}
To meet modern computational demands and to
overcome the fundamental limitations of computing hardware, 
the traditional single-CPU architecture is being replaced by a concurrent system
based on multi-cores or even many-cores.
Therefore, at least until the next technological revolution, the only
way to respond to the growing computing demand is to invest in smarter
concurrent algorithms. 

Synchronization, one of the principal challenges in concurrent programming,
consists in arbitrating concurrent accesses to shared \emph{data structures}:
lists, hash tables, trees, etc.
Intuitively, an efficient data structure must be \emph{highly
  concurrent}:
it should allow  multiple processes to ``make progress'' on it in
parallel.
Indeed, every new implementation of a concurrent data structure is
usually claimed to enable such a parallelism.  
But what does ``making progress'' means precisely?

\myparagraph{Optimal concurrency}
If we zoom in the code of an operation on a typical concurrent data structure, 
we can distinguish \emph{data accesses}, i.e., reads and updates to the
data structure itself, performed as though the operation works on the
data in the absence of concurrency.
To ensure that concurrent operations do not violate 
correctness of the implemented high-level data type (e.g.,
\emph{linearizability}~\cite{Her91} of the implemented \textsf{set}
abstraction),
data accesses are ``protected'' with
\emph{synchronization primitives}, e.g., acquisitions and releases of
locks or atomic read-modify-write instructions like compare-and-swap. 
Intuitively, a process makes progress by performing ``sequential''
data accesses to the shared data, e.g., traversing the data structure
and modifying its content.
In contrast, synchronization tasks, though necessary for correctness,
do not contribute to the progress of an operation.  

Hence, ``making progress in parallel'' can be seen as allowing concurrent
execution of pieces of locally sequential fragments of code.
The more synchronization we use to protect ``critical'' pieces of the
sequential code, the less \emph{schedules}, i.e., interleavings of
data accesses, we accept.  
Intuitively, we would like to use exactly as little synchronization as
sufficient for ensuring linearizability of the high-level implemented abstraction.
This expectation brings up the notion of a \emph{concurrency-optimal}
implementation~\cite{GKR16} that only rejects a schedule if it does
violate linearizability.

To be able to reason about the ``amount of concurrency'' exhibited by
implementations employing different synchronization techniques, we
consider the recently introduced notion of ``local serializability'' (on top of linearizability)
and the metric of the ``amount of concurrency'' defined via sets of
accepted (locally sequential) schedules~\cite{krg16}.
Local serializability, intuitively,  requires the sequence of
sequential steps locally observed by every given process to be
consistent with \emph{some} execution of the sequential algorithm.
Note that these sequential executions can be different for different
processes, i.e., the execution may not be
\emph{serializable}~\cite{Pap79-serial}.
Combined with the standard correctness criterion
of \emph{linearizability}~\cite{HW90,AW04}), local serializability
implies our basic correctness criterion called \emph{LS-linearizability}.
The concurrency properties of LS-linearizable data structures can be
compared on the same level:  implementation $A$ is ``more concurrent'' than
implementation $B$ if the set of schedules accepted by $A$ is a strict
superset of the set of schedules accepted by $B$.  
Thus, a \emph{concurrency-optimal}  implementation accepts \emph{all}
correct (LS-linearizable) schedules.

\myparagraph{A concurrency-optimal binary search tree}
It is interesting to consider binary search trees (BSTs) from the
optimal concurrency perspective, 
as they are believed, as a representative of \emph{search} data structures~\cite{CH98}, to be "concurrency-friendly"~\cite{Sut08}: 
updates concerning different keys are likely to operate on disjoint
sets of tree nodes (in contrast with, e.g., operations on \textsf{queues} or \textsf{stacks}).

We present a novel LS-linearizable concurrent BST-based \textsf{set} implementation. 
We prove that the implementation is optimally concurrent with respect to a
standard internal sequential tree~\cite{krg16}.
The proposed implementation employs the optimistic ``lazy'' locking
approach~\cite{HHL+05} that 
distinguishes \emph{logical} and \emph{physical} deletion of a node and makes sure
that read-only operations are \emph{wait-free}~\cite{Her91}, i.e.,
cannot be delayed by concurrent processes.

The algorithm also offers a few  algorithmic novelties.
Unlike most implementations of concurrent trees, the
algorithm uses multiple locks per node: one lock for the \emph{state} of
the node, and one lock for each of its descendants.
To ensure that only conflicting operations can delay each other, we
use \emph{conditional} read-write locks, where the lock can be
acquired only under certain condition.
Intuitively, only changes in the relevant part of the tree structure may prevent a thread from
acquiring the lock.     
The fine-grained conditional read-write locking of nodes and edges
allows us to ensure that an implementation rejects a schedule only if
it violates linearizability.

\myparagraph{Concurrency optimality and performance}
Of course, optimal concurrency does not necessarily imply performance
nor maximum progress (\`a la \emph{wait-freedom}~\cite{HS11-progress}). 
An extreme example is the transactional memory (TM) data structure. 
TMs typically require restrictions of serializability as a correctness criterion. 
And it is known that rejecting a schedule that is rejected only if it is not serializable 
(the property known as \emph{permissiveness}), requires very heavy local computations~\cite{GHS08-permissiveness,KR11}.
But the intuition is that looking for concurrency-optimal search data structures like trees pays off. 
And this work answers this question in the affirmative by demonstrating empirically that the Java implementation
of our concurrency optimal BST outperforms state-of-the-art BST implementations (\cite{CGR13b,DVY14,BCCO10,EFRB10}) on most workloads.
Apart from the obvious benefit of producing a highly efficient BST, this work suggests that 
optimizing the set of accepted schedules of the sequential code can be an
adequate design principle for building efficient concurrent data structures.

\myparagraph{Roadmap}
The rest of the paper is organized as follows. \cref{sec:alg} describes the details of our BST implementation, starting from the sequential implementation
of \emph{partially-external} binary search tree, our novel conditional read-lock lock abstraction to our concurrency optimal BST implementation.
\cref{sec:opt} formalizes the notion of concurrency optimality and sketches the relevant proofs; complete proofs are delegated to the optional tech report.
\cref{sec:eval} provides details of our experimental methodology and extensive evaluation of our Java implementation. \cref{sec:related} articulates the 
differences with related BST implementations
and presents concluding remarks.

\section{Binary Search Tree Implementation}
\label{sec:alg}
This section consists of two parts. At first, we describe our sequential implementation of the \emph{set} using partially-external
binary search tree. Then, we
build the concurrent implementation
on top of the sequential one by adding synchronization separately for each field of a node.
Our implementation takes only the locks that are neccessary to perform correct modifications of the tree structure.
Moreover, if the field is not going to be modified, the algorithm takes the read lock instead of the write lock. 

We start with the specification of the set type which our binary search tree should
satisfy.
An object of the \emph{set} type stores a set of integer values, initially empty, 
and exports operations $\lit{insert}(v)$, $\lit{remove}(v)$, $\lit{contains}(v)$.
The update operations, $\lit{insert}(v)$ and $\lit{remove}(v)$, return
a boolean response, $\true$ if and only if $v$ is absent (for
$\lit{insert}(v)$) or present (for $\lit{remove}(v)$) in the \emph{set}.  
After $\lit{insert}(v)$ is complete, $v$ is present in the set, and 
after $\lit{remove}(v)$ is complete, $v$ is absent.
The $\lit{contains}(v)$ returns a boolean response, $\true$ if and only if $v$ is present.

A \emph{binary search tree}, later called BST, is a rooted ordered tree in which each node $v$ has a left child and a
right child, either or both of which can be null. The node is named a \emph{leaf}, if it does not have any child.
The order is carried by a value property: the value of each node is stricly greater than the values in its left subtree
and strictly smaller than the values in the right subtree.

\subsection{Sequential implementation}

As for a sequential implementation we chose the well-known \emph{partially-external} binary search tree.
Such tree combines the idea of the internal binary search tree, where the set is represented by
the values from all nodes, and the external binary search tree, where the set is represented by the values
in the leaves while the inner nodes are used for routing (note, that for the external tree the value property
does not consider leafs). The partially-external tree supports two types of nodes: routing and data.
The set is represented by the values contained by the data nodes. To bound the number of routing vertices
by the number of data nodes the tree should satisfy the condition:
all routing nodes have exactly two children.

The pseudocode of the sequential implementation is provided in the Algorithm~\ref{alg:seq}. Here, we give
a brief description. 
The \textbf{traversal} function takes a value $v$ and traverses down the tree from the root
following the corresponding links as long as the current node is not null or its value is not $v$. It returns
the last three visited nodes. The \textbf{contains} function takes a value $v$ and checks the last  node visited
by the traversal and returns whether it is null. The \textbf{insert} function takes a value $v$ and uses the traversal
function to find the place to insert the value. If the node is not null, the algorithm checks whether the node is data or
routing: in the former case it is impossible to insert; in the latter case, the algorithm simply changes the state from routing
to data. If the node is null, then the algorithm assumes that the value $v$ is not in the set and inserts a new node with the
value $v$ as the child of the latest non-null node visited by the traversal function call. The \textbf{delete} function takes a
value $v$ and uses the traversal function to find the node to delete. If the node is null or its state is routing,
the algorithm assumes that the value $v$ is not in the set and finishes. Otherwise, there are three cases depending on
the number of children that the found node has: (i)~if the node has two children, then the algorithm changes its state from data to routing; 
(ii)~if the node has one children, then the algorithm unlinks the node; 
(iii)~finally if the node is a leaf then the algorithm unlinks the node, in addition if the parent is a routing node then it also unlinks 
the parent.

\begin{algorithm*}[!ht]
\caption{Sequential implementation.}
\label{alg:seq}
  \begin{algorithmic}[1]
  	\begin{multicols}{2}
  	{\scriptsize
	
	\Part{Shared variables}{
		\State node is a record with fields: 
		\State $\ms{val}$, its value
		\State $\ms{left}$, its pointer to the left child
		\State $\ms{right}$, its pointer to the right child
		\State $\ms{state} \in \{DATA, ROUTING\}$, its state
		
		\State Initially the tree contains one node $\ms{root}$,
		\State $\ms{root.val} = +\infty$
		\State $\ms{root.state} = DATA$
	}\EndPart
	
	\Statex
	
	\Part{$\lit{traversal}(v)$}{
	\Comment{wait-free traversal}
		\State $\ms{gprev} \gets null; \ms{prev} \gets null$
		\State $\ms{curr} \gets \ms{root}$ \Comment{start from root}
                \While{$\ms{curr} \ne null$}
		  \If{$curr.val = v$}
		      \State break
		  \Else
		    \State $\ms{gprev} \gets \ms{prev}$ 
		    \State $\ms{prev} \gets \ms{curr}$
		    \If{$\ms{curr.val} < v$} 
				\State $\ms{curr} \gets \ms{curr.left}$
			\Else
				\State $\ms{curr} \gets \ms{curr.right}$
			\EndIf
		  \EndIf
		\EndWhile
		
		\Return $\langle \ms{gprev}, \ms{prev}, \ms{curr} \rangle$ \EndReturn 
   	}\EndPart

   	\Statex
	\Part{$\lit{contains}(v)$}{
		\Comment{wait-free contains}
		\State $\langle \ms{gprev}, \ms{prev}, \ms{curr} \rangle \gets \ms{traversal}(v)$
		\Return $\ms{curr} \ne null$~and~$\ms{curr}.state = DATA$ \EndReturn
   	}\EndPart	
	
	\Statex
	\Part{$\lit{insert}(v)$}{
		\State $\langle \ms{gprev}, \ms{prev}, \ms{curr} \rangle \gets traversal(v)$
		
		\If{$\ms{curr} \ne null$} \Comment{node has value $v$ or is a place to insert}
			\State go to Line 8
		\Else 
			\State go to Line 12
		\EndIf
		
		\Return true \EndReturn
   	}\EndPart
	
	\Statex 
	\Part{Update existing node}{
	  \If{$\ms{curr.state} = DATA$}
		\Return false \EndReturn \Comment{$v$ is already in the set}
	  \EndIf
			
	  \State $\ms{curr.state} \gets DATA$
	}\EndPart
	
	\Statex

	\Part{Insert new node}{
	  \State $\ms{newNode.val} \gets v$ \Comment{allocate a new node}

	  \If{$v < \ms{prev.val}$}
		\State $\ms{prev.left} \gets newNode$
	  \Else
		\State $\ms{prev.right} \gets newNode$
	  \EndIf
	}\EndPart
	
	\newpage
	
	\Part{$\lit{delete}(v)$}{
		\State $\langle \ms{gprev}, \ms{prev}, \ms{curr} \rangle \gets traversal(v)$
 		
		\If{$\ms{curr} = null$~or~$\ms{curr.state} \ne DATA$}
			\Return false \EndReturn \Comment{$v$ is not in the set}
		\EndIf

		\If{$\ms{curr}$ has exactly 2 children} 
		  \State go to Line 34
		\EndIf
		
		\If{$\ms{curr}$ has exactly 1 child}
		  \State go to Line 36
		\EndIf
		\If{$\ms{curr}$ is a leaf}
			\If{$\ms{prev.state} = DATA$} 
                          \State go to Line 46
			\Else 
			  \State go to Line 51
			\EndIf
		\EndIf
		
		\Return true \EndReturn
	}\EndPart
	
	\Statex
	\Part{Delete node with two children}{
	  \State $\ms{curr.state} \gets ROUTING$
	}\EndPart
	
	\Statex 
	\Part{Delete node with one child}{
	  \If{$\ms{curr.left} \ne null$}
		\State $\ms{child} \gets \ms{curr.left}$
	  \Else
		\State $\ms{child} \gets \ms{curr.right}$
	  \EndIf

	  \If{$\ms{curr.val} < \ms{prev.val}$}
		\State $\ms{prev.left} \gets \ms{child}$
	  \Else
		\State $\ms{prev.right} \gets \ms{child}$
	  \EndIf
	}\EndPart
	
	\Statex 
	\Part{Delete leaf with DATA parent}{
	  \If{$\ms{curr}$ is left child of $\ms{prev}$}
	    \State $\ms{prev.left} \gets null$
	  \Else
		\State $\ms{prev.right} \gets null$
	  \EndIf
	}\EndPart
	
	\Statex 
	\Part{Delete leaf with ROUTING parent}{
		\State \Comment{save second child of $\ms{prev}$ into $\ms{child}$}
		  \If{$\ms{curr}$ is left child of $\ms{prev}$} 
			\State $\ms{child} \gets \ms{prev.right}$
		  \Else
			\State $\ms{child} \gets \ms{prev.left}$
		  \EndIf

		  \If{$\ms{prev}$ is left child of $\ms{gprev}$}
			\State $\ms{gprev.left} \gets \ms{child}$
		  \Else                         
			\State $\ms{gprev.right} \gets \ms{child}$
		  \EndIf
	}\EndPart

	}
	\end{multicols}
  \end{algorithmic}
\end{algorithm*}

\subsection{Concurrent implementation}
As the basis of our concurrent implementation we took the idea of optimistic algorithms, where
the algorithm reads all necessary variables without synchronizations
and right before the modification, the algorithm
takes all the locks and checks the consistency of all the information it read. 
As we show in the next section, we build upon the partially-external property of the BST to 
provide a concurrency-optimal BST.
Let us first give more details on how the algorithm is implemented.

\myparagraph{Field reads} Since our algorithm is optimistic we do not want to read the same field twice.
To overcome this problem when the algorithm reads the field it stores it in ``cache'' and the further accesses
return the ``cached'' value.
For example, the reads of the $left$ field in Lines 28 and 29 of Algorithm~\ref{alg:concur} return the same (cached) value.

\myparagraph{Deleted mark} As usual in concurrent algorithms with wait-free traversals,
the deletion of the node happens in two stages. At first, the delete operation
logically removes a node from the tree by setting the boolean flag to \textbf{deleted}. Secondly,
the delete operation updates the links to physically remove the node. By that, any traversal that
suddenly reaches the ``under-deletion'' node, sees the deletion node and could restart the operation.

\myparagraph{Locks} In the beginning of the section we noted that we have locks separately for each field of a node
and the algorithm takes only the necessary type of lock: read or write.
For that, we implemented read-write lock simply as one $lock$ variable. The smallest bit
of $lock$ indicates whether the write lock is taken or not, the rest part of the variable indicates the
number of readers that have taken a lock. In other words, $lock$ is zero if the lock is not taken,
$lock$ is one if the write lock is taken, otherwise, $lock$ divided by two represents the number of times
the read lock is taken. The locking and unlocking are done using the atomic compare-and-set primitive.
Along, with standard $\lit{tryWriteLock}$, $\lit{tryReadLock}$, $\lit{unlockWrite}$ and $\lit{unlockRead}$
we provide additional six functions on a node: $\lit{tryLockLeftEdge(Ref|Val)(exp)}$,
$\lit{lockRightEdge(Ref|Val)(exp)}$ and $\lit{try(Read|Write)LockState(exp)}$
(Starting from here, we use the notation of bar to not duplicate the similar names;
such notation should be read as either we choose the first option or the second option.)

Function $\lit{tryLock(Left|Right)EdgeRef}$ ensures that the lock is taken
only if the field ($left$ or $right$) guarded by that lock is equal to $exp$, i.e., the child node has
not changed, and the current node is not deleted, i.e., its deleted mark is not set.
Function $\lit{tryLock(Left|Right)EdgeVal}$ ensures that the lock is taken
only if the value of the node in the field ($left$ or $right$) guarded by that lock is equal to $exp$
, i.e., the node could have changed by the value inside does not,
and the current node is not deleted, i.e., its deleted mark is not set.
Function $\lit{try(Read|Write)LockState(exp)}$ ensures that the lock is taken only if the value of the $state$
is equal to $exp$ and the current node is not deleted, i.e., its deleted mark is not set.

These six functions are implemented in the same manner: the function reads necessary fields and lock variable, checks
the conditions, if successful it takes a corresponding lock, then checks the conditions again, if unsuccessful
it releases lock. In most cases in the pseudocode we used a substitution $tryLockEdge(Ref|Val)(node)$ instead of
$tryLock(Left|Right)Edge(Ref|Val)(exp)$. This substitution, given not-null value, decides whether the $node$
is the left or right child of the current node and calls the corresponding function providing $node$ or $node.value$.

\begin{algorithm*}[!h]
\caption{Concurrent implementation.}
  \label{alg:concur}
  \begin{algorithmic}[1]
  	\begin{multicols}{2}
  	{\scriptsize

  	\Part{$\lit{contains}(v)$}{
  	  \State $\langle \ms{gprev}, \ms{prev}, \ms{curr} \rangle \gets traversal(v)$
  	  \Return $\ms{curr \ne null} \wedge \ms{curr.state} = DATA$ \EndReturn
  	}\EndPart

	\Statex	             

	\Part{$\lit{insert}(v)$}{
		\State $\langle \ms{gprev}, \ms{prev}, \ms{curr} \rangle \gets traversal(v)$
		
		\If{$\ms{curr} \ne null$}
		  \State go to Line~\ref{line:ins:start}
		\Else
 		  \State go to Line~\ref{line:ins-null:start}
		\EndIf
		
		\State Release all locks
		\Return $\true$ \EndReturn

	\Statex
	
   	\Statex \textbf{Update existing node}
	  \If{$\ms{curr.state} = DATA$} \label{line:ins:start}
	    \Return $\false$ \EndReturn
	  \EndIf
	  
	  \State $\ms{curr.tryWriteLockState(ROUTING)}$ \label{line:ins:curr} 
	  \State $\ms{curr.state} \gets DATA$ \label{line:ins:lin}

	\Statex
   	
   	\Statex \textbf{Insert new node}
   	  \State $\ms{newNode.val} \gets v$ \label{line:ins-null:start} 
	   
	  \If{$v < \ms{prev.val}$}
            \State $prev.tryLockLeftEdgeRef(null)$ \label{line:ins-null:prev:null:L}
            \State $\ms{prev.slock.tryReadLock()}$ \label{line:ins-null:prev:L}
            \If{$\ms{prev.deleted}$} \label{line:ins-null:prev:deleted:L}
              \State Restart operation
            \EndIf
	    \State $\ms{prev.left} \gets newNode$ \label{line:ins-null:lin:L}
	  \Else
	    \State $prev.tryLockRightEdgeRef(null)$ \label{line:ins-null:prev:null:R}
            \State $\ms{prev.slock.tryReadLock()}$ \label{line:ins-null:prev:R}
            \If{$\ms{prev.deleted}$} \label{line:ins-null:prev:deleted:R}
              \State Restart operation
            \EndIf

            \State $\ms{prev.right} \gets newNode$ \label{line:ins-null:lin:R}
	  \EndIf
	}\EndPart

	\Statex

	\Part{$\lit{delete}(v)$}{
	  \State $\langle \ms{gprev}, \ms{prev}, \ms{curr} \rangle \gets traversal(v)$
      \Statex \Comment{All restarts are from this Line}

	  \If{$\ms{curr} = null \vee curr.state \ne DATA$}
	    \Return false \EndReturn
      \EndIf	 

	  \If{$\ms{curr}$ has exactly 2 children}
	    \State go to Line~\ref{line:del-2:curr:data}
	  \EndIf
	  
	  \If{$\ms{curr}$ has exactly 1 child}
	    \State go to Line~\ref{line:del-1:start}
	  \EndIf
	  
	  \If{$\ms{curr}$ is a leaf}
	    \If{$\ms{prev}.state = DATA$}
	      \State go to Line~\ref{line:del-0:data:start}
	    \Else
	      \State go to Line~\ref{line:del-0:routing:start}
	    \EndIf
	  \EndIf
	    
      \State Release all locks 
      \Return $\true$ \EndReturn	  
	
	\Statex
	
	\Statex \textbf{Delete node with two children}
	  \State $\ms{curr.tryWriteLockState(DATA)}$ \label{line:del-2:curr:data}

	  \If{$\ms{curr}$ does not have 2 children} \label{line:del-2:curr:2children}
	     \State Restart operation
	  \EndIf
	  \State $curr.state \gets ROUTING$ \label{line:del-2:lin}

	\newpage

	\Statex \textbf{Lock acquisition routine for vertex with one child}
          \State $curr.tryLockEdgeRef(child)$ \label{line:del-1:curr:child}
          \State $prev.tryLockEdgeRef(curr)$ \label{line:del-1:prev:curr}

	  \State $curr.tryWriteLockState(DATA)$ \label{line:del-1:curr:data}
	  \If{$\ms{curr}$ has 0 or 2 children} \label{line:del-1:curr:1child}
	    \State Restart operation
	  \EndIf
	                                                                            
	\Statex

	\Statex \textbf{Delete node with one child}
	  \If{$\ms{curr.left} \ne null$} \label{line:del-1:start} \label{line:del-1:cache-1}
		  \State $\ms{child} \gets \ms{curr.left}$ \label{line:del-1:cache-2}
	  \Else
		  \State $\ms{child} \gets \ms{curr.right}$
	  \EndIf

	  \If{$\ms{curr.val} < \ms{prev.val}$} 
	    \State perform lock acquisition at Line~\ref{line:del-1:curr:child}
	    \State $\ms{curr.deleted} \gets \true$ \label{line:del-1:ldel:L}
	    \State $\ms{prev.left} \gets \ms{child}$ \label{line:del-1:lin:L}
	  \Else
	    \State perform lock acquisition at Line~\ref{line:del-1:curr:child}
	    \State $\ms{curr.deleted} \gets \true$ \label{line:del-1:ldel:R}
	    \State $\ms{prev.right} \gets \ms{child}$ \label{line:del-1:lin:R}
	  \EndIf

    \Statex

	\Statex \textbf{Lock acquisition routine for leaf}
	  \State $prev.tryLockEdgeVal(curr)$ \label{line:del-0:prev:curr}
	  \If{$\ms{v} < \ms{prev.key}$} \Comment{get current child}
	    \State $\ms{curr} \gets prev.left$
	  \Else
	    \State $\ms{curr} \gets prev.right$
	  \EndIf
	  \State $\ms{curr.tryWriteLockState(DATA)}$ \label{line:del-0:curr}
	  \If{$\ms{curr}$ is not a leaf} \label{line:del-0:curr:leaf}
	    \State Restart operation
	  \EndIf

        \Statex

	\Statex \textbf{Delete leaf with DATA parent}
	  \If{$\ms{curr.val} < \ms{prev.val}$} \label{line:del-0:data:start}
	    \State perform lock acquisition at Line~\ref{line:del-0:prev:curr}
	    \State $\ms{prev.tryReadLockState(DATA)}$ \label{line:del-0:data:prev:L} 
	    \State $\ms{curr.deleted} \gets \true$ \label{line:del-0:data:ldel:L}
	    \State $\ms{prev.left} \gets null$ \label{line:del-0:data:lin:L}
	  \Else
            \State perform lock acquisition at Line~\ref{line:del-0:prev:curr}
            \State $\ms{prev.tryReadLockState(DATA)}$ \label{line:del-0:data:prev:R}
	    \State $\ms{curr.deleted} \gets \true$ \label{line:del-0:data:ldel:R}
	    \State $\ms{prev.right} \gets null$ \label{line:del-0:data:lin:R}
	  \EndIf
	
	\Statex 
	
	\Statex \textbf{Delete leaf with ROUTING parent}
	  \If{$\ms{curr.val} < \ms{prev.val}$} \label{line:del-0:routing:start}
	    \State $\ms{child} \gets \ms{prev.right}$
	  \Else
	    \State $\ms{child} \gets \ms{prev.left}$
	  \EndIf	    
	
	  \If{$\ms{prev}$ is left child of $\ms{gprev}$}
	    \State perform lock acquisition at Line~\ref{line:del-0:prev:curr}
            \State $prev.tryEdgeLockRef(child)$ \label{line:del-0:routing:prev:child:L}
            \State $gprev.tryEdgeLockRef(prev)$ \label{line:del-0:routing:gprev:prev:L}
	    \State $\ms{prev.tryWriteLockState(ROUTING)}$ \label{line:del-0:routing:prev:L}
	    \State $\ms{prev.deleted} \gets \true$ \label{line:del-0:routing:ldel:prev:L}
	    \State $\ms{curr.deleted} \gets \true$ \label{line:del-0:routing:ldel:curr:L}
	
	    \State $\ms{gprev.left} \gets \ms{child}$ \label{line:del-0:routing:lin:L}
	  \Else
	    \State perform lock acquisition at Line~\ref{line:del-0:prev:curr}
            \State $prev.tryEdgeLockRef(child)$ \label{line:del-0:routing:prev:child:R}
            \State $gprev.tryEdgeLockRef(prev)$ \label{line:del-0:routing:gprev:prev:R}
	    \State $\ms{prev.tryWriteLockState(ROUTING)}$ \label{line:del-0:routing:prev:R}
	    \State $\ms{prev.deleted} \gets \true$ \label{line:del-0:routing:ldel:prev:R}
	    \State $\ms{curr.deleted} \gets \true$ \label{line:del-0:routing:ldel:curr:R}

	    \State $\ms{gprev.right} \gets \ms{child}$ \label{line:del-0:routing:lin:R}
	  \EndIf

	}\EndPart

	}
	\end{multicols}
  \end{algorithmic}
\end{algorithm*}

\vspace{-0.7em}
\section{Concurrency optimality and correctness}
\label{sec:opt}

\vspace{-0.7em}
In this section, we show that our implementation is
\emph{concurrency-optimal}~\cite{GKR16}.
Intuitively, a concurrency-optimal implementation employs as much
synchronization as necessary for ensuring correctness of the
implemented high-level abstraction~--- in our case, the linearizable set
object~\cite{Her91}.

Recall our \emph{sequential} BST implementation and imagine that we run it in a \emph{concurrent} environment.
We refer to an execution of  this concurrent algorithm as a
\emph{schedule}.
A schedule thus consists of reads, writes, node creation events, and invocation and responses
of high-level operations.

Notice that in every such schedule, any operation witnesses a
\emph{consistent} tree state locally, i.e., it cannot distinguish the
execution from a sequential one.   
It is easy to see that the local views \emph{across operations} may
not be mutually consistent, and this simplistic concurrent algorithm is not linearizable.
For example, two insert operations that concurrently traverse the
tree may update the same node so that one of the operations
``overwrites'' the other (so called the ``lost update'' problem).
To guarantee linearizability, one needs to ensure that only correct
(linearizable) schedules are accepted.
In this section, we are going to show first that this is indeed the
case with our algorithm: all the schedules it \emph{accepts} are correct.
More precisely, a schedule $\sigma$ is accepted by an algorithm if it
has an execution in which the sequence of high-level invocations and
responses, reads, writes, and node creation events (modulo the
restarted fragments) is $\sigma$~\cite{krg16}.    

Further, we show that, in a strict sense, our algorithm accepts
\emph{all} correct schedules. 
In our definition of correctness, we demand that at all times
the algorithm maintains a \emph{BST} that does not
contain nodes that were previously \emph{physically deleted}. 
Formally, a set of nodes reachable
from the \textit{root} is a \emph{BST} if:        
(i) they form a tree rooted at node \textit{root};   
(ii) this tree satisfies the \emph{value property}: for each node with value $v$ all the values in the left subtree
are less than $v$ and all the values in the right subtree are bigger than $v$;
(iii) each routing node in this tree has two children.

Now we say that a schedule is \emph{observably correct} if
each of its prefixes $\sigma$ satisfies the following conditions:
(i) subsequence of high-level invocations and responses of operations that made a write in $\sigma$
has a linearization with respect to the \textsf{set} type.
(ii) the data structure after performing $\sigma$ is a BST $B$;
(iii) $B$ does not contain a node $x$ such that there exist 
$\sigma'$ and $\sigma''$, such that $\sigma'$ is a prefix of
  $\sigma''$, $\sigma''$ is a prefix of $\sigma$, $x$ is in the
  BST after $\sigma'$, and $x$ is not in the BST after $\sigma''$.
(iii)    has a linearization

\begin{theorem}[Correctness]\label{th:corr}
The schedule corresponding to any execution of our BST implementation is
observably correct.
\end{theorem}
Finally, we say that an implementation is \emph{concurrency-optimal} if it accepts all observably correct schedules.

\vspace{-0.2em}
\begin{theorem}[Optimality]\label{th:opt}
  Our BST implementation is concurrency-optimal.
\end{theorem}
\vspace{-0.2em}

The intuition behind the proof of Theorem~\ref{th:opt} is the following. 
We show that for each observably correct schedule there exists a
matching  execution of our implementation.
Therefore, only schedules not observably correct can be rejected by
our algorithm.
The construction of an execution that matches an observably correct schedule is possible, in particular, due to
the fact that every critical section in our algorithm  contains
exactly one event of the schedule.
Thus, the only reason to reject a schedule is that some condition on a
critical section does not hold and, as a result, the operation must be
restarted.
By accounting for all the conditions under which an operation
restarts, we show that this may only happen if, otherwise, the
schedule violates observable correctness.

\begin{figure*}[!t]
	\begin{center}
	\subfloat[Scenario depicting the concurrent execution of $\lit{insert}(1)$ and $\lit{insert}(3)$; rejected by popular BSTs like \cite{CGR13b,DVY14,BCCO10,EFRB10}, 
	but must be accepted by a concurrency-optimal BST\label{sfig:ins1}]{\scalebox{.58}[.58]{\includegraphics[]{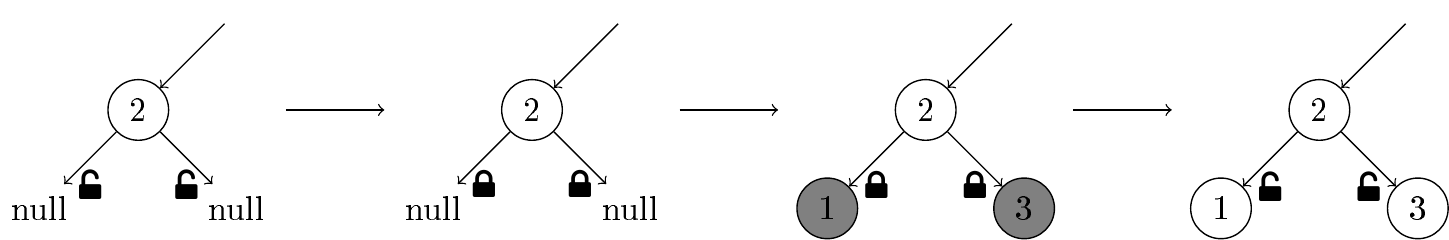}}}
	\end{center}
        \vspace{-2em}
        \begin{center}
	\subfloat[Scenario depicting the concurrent execution of $\lit{delete}(3)$ that is concurrent with a second $\lit{delete}(3)$ followed by a successful $\lit{insert}(3)$; 
	rejected by all the popular BSTs~\cite{CGR13b,DVY14,BCCO10,EFRB10, NM14}, but must be accepted by a concurrency-optimal BST\label{sfig:del1}]{\scalebox{.55}[.55]{\includegraphics[]{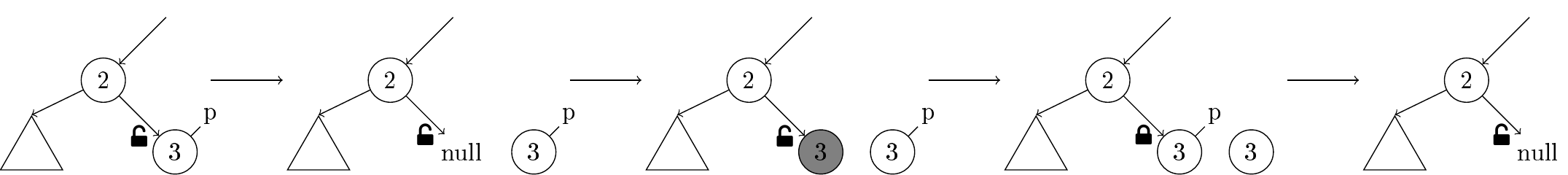}}}
	\end{center}
	\caption{Examples schedules rejected by concurrent BSTs not concurrency-optimal}
        \label{fig:sch}
        \vspace{-2em}
\end{figure*}
\myparagraph{Suboptimality of related BST algorithms}
To understand the hardness of building linearizable concurrency optimal BSTs, we explain how some typical correct schedules are rejected by current state-of-the-art BST
algorithms against which we evaluate the performance of our algorithm.
Consider the concurrency scenario depicted in Figure~\ref{sfig:ins1}. There are two concurrent operations $\lit{insert}(1)$ and
$\lit{insert}(3)$ performed on a tree.
They traverse to the corresponding links (part a)) and lock them concurrently (part b)). Then they insert new nodes (part c)).
Note that this is a correct schedule of events; however, most BSTs including the ones we compare our implementation against~\cite{CGR13b,DVY14,BCCO10,EFRB10} reject
this schedule or similar.
However, using multiple locks per node allows our concurrency-optimal implementation to accept this schedule.

The second schedule is shown in the Figure~\ref{sfig:del1}. There is one operation $p = \lit{delete}(3)$ performed on a tree
shown in part a). It traverses to a node $v$ with value $3$. Then, some concurrent operation $\lit{delete}(3)$ unlinks node $v$ (part b)).
Later, another concurrent operation inserts a new node with value $3$ (part c)).
Operation $p$ wakes up and locks a link since the value $3$ is the same (part d)).
Finally, $p$ unlinks the node with value $3$ (part e)). Note that this is a correct schedule since both the delete operations can be successful; however, all the 
BSTs we are aware of reject this schedule or similar~\cite{CGR13b,DVY14,BCCO10,EFRB10, NM14}. While, there is an execution of our concurrency-optimal BST that accepts this schedule.

\section{Implementation and evaluation}
\label{sec:eval}

\myparagraph{Experimental setup} For our experiments
we used two machines to evaluate the versioned binary search tree.
The first is a 4-processor Intel Xeon E7-4870 2.4 GHz server (Intel)
with 20 threads per processor (yielding 80 hardware threads in total), 512 Gb of RAM, running Fedora 25.
This machine has Java 1.8.0\_111-b14 and HotSpot VM 25.111-b14.
Second machine is a 4-processor AMD Opteron 6378 2.4 GHz server (AMD)
with 16 threads per processor (yielding 64 threads in total), 512 Gb of RAM, running Ubuntu 14.04.5.
This machine has Java 1.8.0\_111-b14 and HotSpot JVM 25.111-b14.

\myparagraph{Binary Search Tree Implementations} We compare our algorithm,
denoted as Concurrency Optimal or CO, against four other implementations of concurrent BST.
They are:
\begin{enumerate*}[label={\arabic*)}]
\item the lock-based contention-friendly tree by Crain et al. (\cite{CGR13b}, Concurrency Friendly or CF),
\item the lock-based logical ordering AVL-tree by Drachsler et al. (\cite{DVY14}, Logical Ordering or LO)
\item the lock-based tree by Bronson et al. (\cite{BCCO10}, BCCO) and
\item the lock-free tree by Ellen et al. (\cite{EFRB10}, EFRB).
\end{enumerate*}
All these implementations are written in Java and taken
from the {\texttt{synchrobench repository}}~\cite{Gra15}.
In order to make the comparison equitable, we remove rotation routines from
the CF-, LO- and CO- trees implementations. We are aware of efficient
lock-free tree by Natarajan and Mittal (\cite{NM14}), but unfortunately
we were unable to find it written on Java.

\myparagraph{Experimental methodology} For our experiments, we use the environment provided by the synchrobench
library. To compare the performance we considered the following parameters:
\begin{itemize}
\item \textbf{Workloads.} Each workload distribution is characterized by the percent $x\%$
of update operations. This means that the tree will be requested to make
$100 - x\%$ of \texttt{contains} calls, $x / 2 \%$ of \texttt{insert} calls and $x / 2 \%$ of \texttt{delete} calls.
We considered three different workload distributions: 0\%, 20\% and 100\%.
\item \textbf{Tree size.} On the workloads described above, the tree size depends on the size of the key space
(the size is approximately half of the range).
We consider three different key ranges: $2^{15}$, $2^{19}$ and $2^{21}$.
To ensure consistent results, rather than starting with an empty tree, we pre-populated the tree before
execution.
\item \textbf{Degree of contention.} This depends on the number of threads in a machine.
We take enough points to reason about the behaviour of curves.
\end{itemize}
In fact, we made experiments on a larger number of settings but we shortened our
presentation due to lack of space.
We chose the settings such that we had two extremes and one middle point. For workload, we chose 20\% of attempted updates as a middle point,
because it corresponds to real life situation in database management where the percentage of successful updates is 10\%. (In our testing environment
we expect only half of update calls to succeed)

\begin{figure*}[!t]
    \center
	\subfloat[Evaluation of BST implementations on Intel\label{sfig:results:intel}]{\scalebox{0.58}[0.58]{\includegraphics[width=1\linewidth, height=.75\textheight]{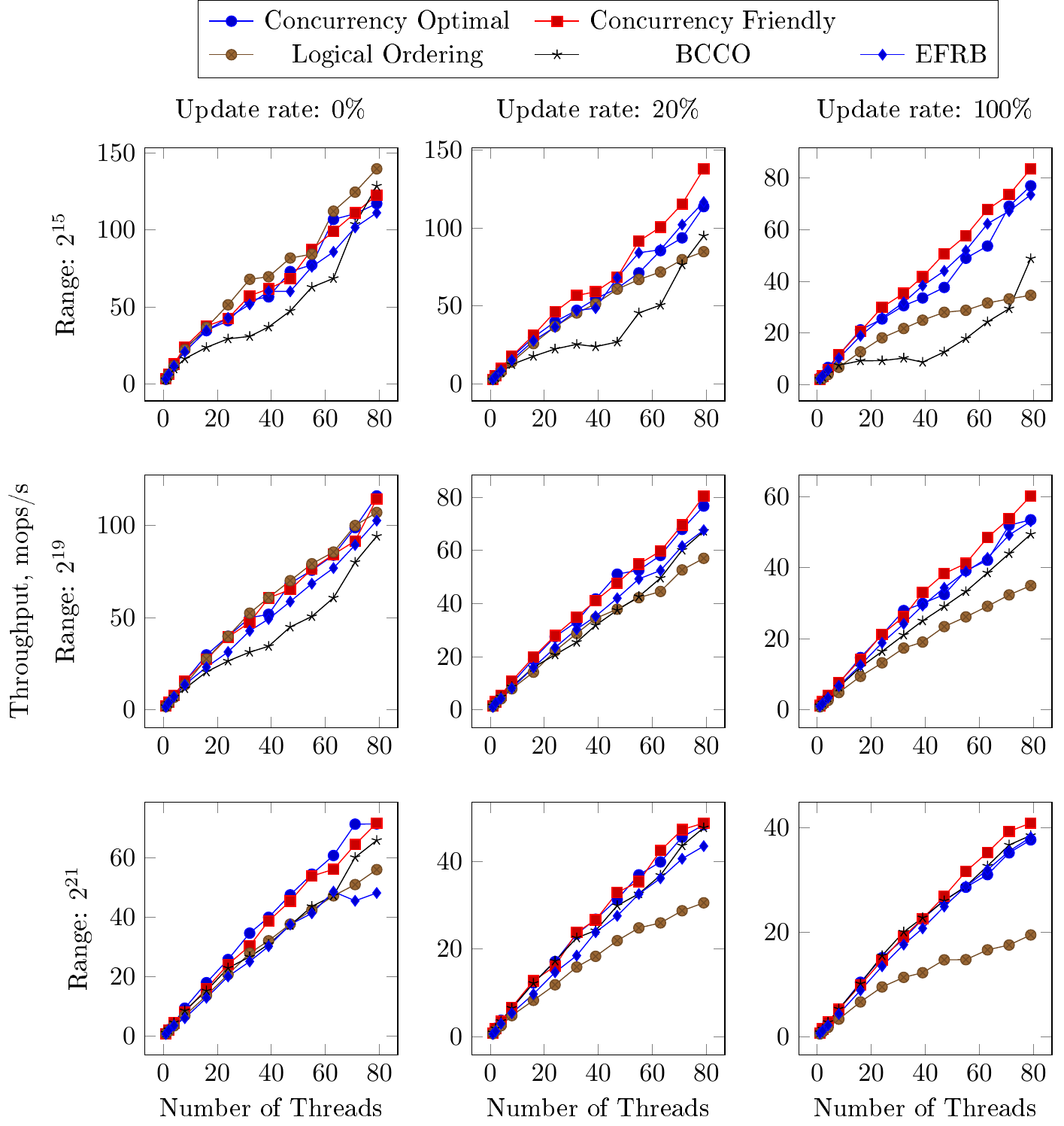}}}
        \\
        \vspace{.5mm}        
	\subfloat[Evaluation of BST implementations on AMD\label{sfig:results:amd}]{\scalebox{0.58}[0.58]{\includegraphics[width=1\linewidth, height=.75\textheight]{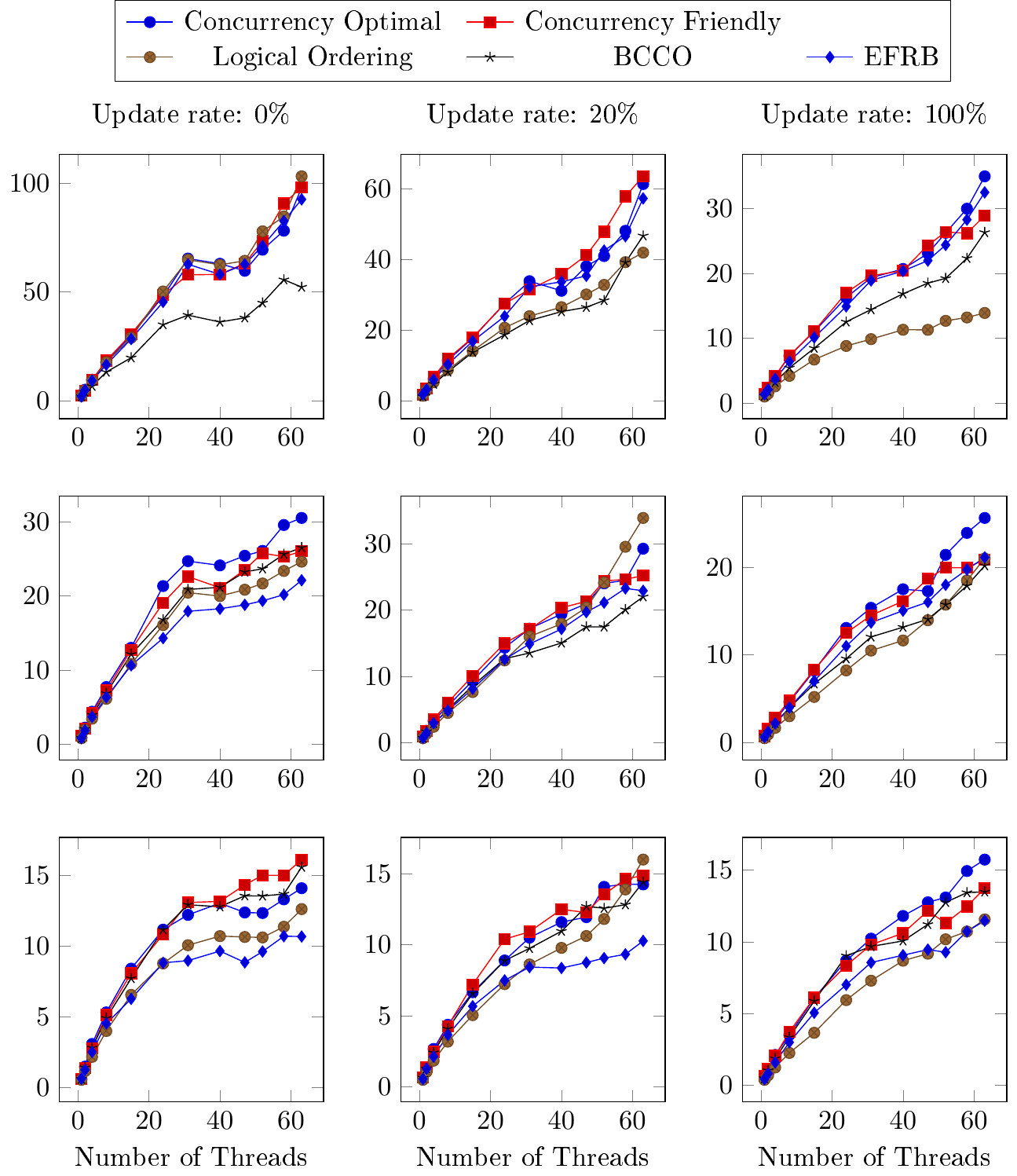}}}
	\caption{Performance evaluation of concurrent BSTs}
        \label{fig:eval}
\end{figure*}

\myparagraph{Results}
To get meaningful results we average through up to 25 runs. Each run
is carried out for 10 seconds with a warmup of 5 seconds.
Figure \ref{sfig:results:intel} (and resp. \ref{sfig:results:amd}) contains the
results of executions on Intel (and resp. AMD) machine.
It can be seen that with the increase of the size the performance of our algorithm becomes better relatively to CF-tree.
This is due to the fact that with bigger size the cleanup-thread in CF-tree implementation
spends more time to clean the tree out of logically deleted vertices, thus,
the traversals has more chances to pass over deleted vertices, leading to longer traversals.
By this fact and the trend shown, we could assume that CO-tree outperforms CF-tree on bigger sizes.
On the other hand, BCCO-tree was much worse on $2^{15}$ and became similar to CO-tree on $2^{21}$. This happened because the races
for the locks become more unlikely. This helped much to BCCO-tree, because
it uses high-grained locking. Since, our algorithm is ``exactly'' the same
without order of locking, On bigger sizes we could expect that our implementation will continue to perform
similarly to CO-tree, because the difference in CO-tree and CF-tree implementations is only in grabbing locks method.
By that, we could state that our algorithm works well not depending on the size.
As the percentage of contains operations increases, the difference between our algorithm and CF-tree becomes smaller, moreover,
our algorithm seems to perform better than other trees.

\section{Related work and Discussion}
\label{sec:related}
Although lots of efforts have been devoted to improve the scalability of binary search trees as the concurrency level increases, we are not aware of any work measuring their concurrency.
Measuring concurrency has already been expressed as the idea of comparing a concurrent data structure to its sequential counterpart~\cite{PODCKRV12} and this reasoning was generalized 
to a class of search structures organizing nodes in a directed acyclic graph~\cite{GKR16}. While a similar reasoning was applied to a linked list~\cite{optimistic-list15-disc} to measure its concurrency, 
the existence of a concurrency-optimal binary search tree has never been addressed to our knowledge.

The transactional red-black tree~\cite{CCKO08} uses software transactional memory 
without sentinel nodes to limit conflicts between concurrent transactions but restarts the update operation after its rotation aborts.
Optimistic synchronization as seen in transactional memory was used to implement 
a practical lock-based binary search tree~\cite{BCCO10}.
The speculation-friendly tree~\cite{CGR12} is a partially internal binary search tree that marks internal nodes as logical deleted to reduce conflicts between software transactions. It decouples structural operation from abstract operations to  rebalance when contention disappears. 
Some red-black trees were optimized for hardware transactional memory and compared with bottom-up and top-down fine-grained locking techniques~\cite{SNGK15}.
The contention-friendly tree~\cite{CGR13b} is a lock-based partially-external binary search tree that provides lock-free lookups that rebalances  when contention disappears.
The logical ordering tree~\cite{DVY14} combines the lock-free lookup with on-time removal during deletes.

The first lock-free tree proposal~\cite{EFRB10} uses a single-word CAS to implement a non-blocking binary search trees and does not rebalance. 
%
%
Howley and Jones~\cite{HJ12} proposed an internal lock-free binary search tree where each node keeps track of the operation currently modifying it. 
Chatterjee et al.~\cite{CNT14} proposed a lock-free binary search tree, but we are not aware of any 
implementation. 
Natarajan and Mittal~\cite{NM14} proposed an efficient lock-free binary search tree implementation that uses edge markers.
It outperforms both the lock-free binary search trees from Howley and Jones~\cite{HJ12} and Ellen et al.~\cite{EFRB10}.
%

\bibliography{references}
 \section{Proof of correctness}
\label{sec:correctness}

In general, the correctness of the parallel algorithm is carried by the proofs of linearizability and
deadlock-freedom. In our paper we add additional constraints on the possible executions of our algorithm:
they have to carry the observably correct schedules. We consider the schedule to be observably correct if it satisfies three conditions:
the prefix of the schedule is linearizable; at any time the tree is a BST; and the algorithm never links the unlinked node back.
This notion could be formally defined as follows.

\begin{definition}
A schedule is \textit{observably correct} if each of its prefixes $\sigma$ satisfies the following conditions:
\begin{itemize}
\item subsequence of high-level invocations and responses of operations that made a write in $\sigma$
has a linearization with respect to the \lit{set} type;
\item the data strucure after performing $\sigma$ is a BST $B$;
\item BST after performing $\sigma$ does no contain a node $x$ such that there exist $\sigma'$ and $\sigma''$, such that
$\sigma'$ is a prefix of $\sigma''$, $\sigma''$ is a prefix of $\sigma$, $x$ is in the BST after $\sigma'$, and $x$ is not in the BST
after $\sigma''$.
\end{itemize}
\end{definition}

The theorem about the correctness of the algorithm could be stated as follows.

\begin{theorem}
The algorithm is correct if:
\begin{itemize}
\item the schedule corresponding to any execution of the algorithm is observably correct.
\item the algorithm is deadlock-free.
\end{itemize}
\end{theorem}

We split our proof into three parts: the structural properties, i.e., the tree is a BST and an unlinked node cannot be linked back,
the linearizability and deadlock-freedom. 

\subsection{Structural correctness}
At first, we prove that our search tree satisfies the structural properties at any point in time, i.e.,
the second and the third property of observably correctness.
Later we refer to these properties as Properties 1, 2, 3 and 4.
\begin{theorem}
The following properties are satisfied at any point of time during the execution:
\begin{itemize}
\item The value property of BST is preserved.
\item Every routing node has two children.
\item Any non-physically deleted node is reachable from the root.
\item Any physically deleted node is non-reachable from the root.
\end{itemize}
\end{theorem}
\begin{proof}
The first two properties are non-trivial by themselves, but we could refer to papers \cite{BCCO10} and \cite{CGR13b} that
use the similar partially-external algorithm.

The last two properties follows in a straightforward way from the fact that
during physical deletion the algorithm takes locks.
\end{proof}

\subsection{Linearizability}
To prove the linearizability of our algorithm, we need to define the linearization points of
\lit{insert}, \lit{delete} and \lit{contains} operations.
When defined the linearization points it could be straightforwardly seen that if the execution
is linearizable then each prefix of the corresponding schedule is linearizable. So, for us, it will
be enough just to prove that any execution is linearizable.

\textbf{High-level histories and linearizability.}
A \textit{high-level history} $\widetilde{H}$ of an execution $\alpha$ is the subsequence of $\alpha$
consisting of all invocations and responses of (high-level) operations.

A complete high-level history $\widetilde{H}$ is \textit{linearizable} with respect to an object type $\tau$ 
if there exists a sequential high-level history $S$ equivalent to $\widetilde{H}$ such that 
\begin{enumerate}
  \item $\to_{\widetilde{H}} \subseteq \to_S$
  \item $S$ is consistent with the sequential specification of type $\tau$.
\end{enumerate}

Now a high-level history $\widetilde{H}$ is linearizable if it can be \textit{completed} (by adding matching responses to a subset
of incomplete operations in $\widetilde{H}$ and removing the rest) to a linearizable high-level history.

\textbf{Completions.} We obtain a completion $\tilde{H}$ of history $H$ as follows. The invocation of an incomplete
contains operation is discarded. The invocation of an incomplete $\pi = \lit{insert}$ operation that has not performed
a write at Lines \ref{line:ins:curr}, \ref{line:ins-null:lin:L} (\ref{line:ins-null:lin:R}) of the Algorithm~\ref{alg:concur} are discarded;
otherwise, $\pi$ is completed with the response $\true$.
The invocation of an incomplete $\pi = \lit{delete}$ operation that has
not performed a write at Lines \ref{line:del-2:lin}, \ref{line:del-1:lin:L} (\ref{line:del-1:lin:R}),
\ref{line:del-0:data:lin:L} (\ref{line:del-0:data:lin:R}), \ref{line:del-0:routing:lin:L} (\ref{line:del-0:routing:lin:R})
of the Algorithm~\ref{alg:concur} is discarded;
otherwise, it is completed with the response $\true$.

Note, that the described completions correspond to the completions in which the
completed operations made at least write of the sequential algorithm.

\textbf{Linearization points.} We obtain a sequential high-level history $\tilde{S}$ equivalent to $\tilde{H}$ 
by associating a linearization point $l_{\pi}$ with each operation $\pi$.
In some cases, our choice of the linearization point depends on the time interval between the invocation and the response of
 the execution of $\pi$, later referred to as the interval of $\pi$. 
For example, the linearization point of $\pi$ in the timeline should lie in the interval of $\pi$.

Below we specify the linearization point of the operation $\pi$ depending on its type.

\textbf{Insert.}
For $\pi = \lit{insert}(v)$ that returns $\true$, we have two cases:
\begin{enumerate}
\item A node with key $v$ was found in the tree. Then $l_{\pi}$ is associated with the write in Line \ref{line:ins:lin}
of the Algorithm~\ref{alg:concur}.
\item A node with key $v$ was not found in the tree. 
Then $l_{\pi}$ is associated with the writes in Lines \ref{line:ins-null:lin:L} or \ref{line:ins-null:lin:R} of the Algorithm~\ref{alg:concur},
depending on whether the inserted node is left or right child.
\end{enumerate}
For $\pi = \lit{insert}(v)$ that returns $\false$, we have three cases:
\begin{enumerate}
\item If there exists a successful $\lit{insert}(v)$ whose linearization point lies in the interval of $\pi$,
then we take the first such $\pi'= \lit{insert(v)}$ and linearize right after $l_{\pi'}$.
\item If there exists a successful $\lit{delete}(v)$ whose linearization point lies in the interval of $\pi$,
then we take the first such $\pi' = \lit{delete(v)}$ and linearize right before $l_{\pi'}$.
\item Otherwise, $l_{\pi}$ is the call point of $\pi$.
\end{enumerate}

\textbf{Delete.}
For $\pi = \lit{delete}(v)$ that returns $\true$ we have four cases, 
depending on the number of children of the node with key $v$, i.e., the node $\ms{curr}$:
\begin{enumerate}
\item $\ms{curr}$ has two children. Then $l_{\pi}$ is associated with the write in Line \ref{line:del-2:lin} of the Algorithm~\ref{alg:concur}.
\item $\ms{curr}$ has one child. Then $l_{\pi}$ is associated between the writes in
Line \ref{line:del-1:ldel:L} (\ref{line:del-1:ldel:R}) and in Line \ref{line:del-1:lin:L} (\ref{line:del-1:lin:R})
of the Algorithm~\ref{alg:concur}, depending on whether $\ms{curr}$ is left or right child.
The exact position is calculated as what comes last: Line \ref{line:del-1:ldel:L} (\ref{line:del-1:ldel:R}) or
the last invocation of unsuccessful $\lit{insert}(v)$ or $\lit{contains}(v)$
that reads the node $curr$.
\item $\ms{curr}$ is a leaf with a data parent. Then $l_{\pi}$ is associated between the writes in
Line \ref{line:del-0:data:ldel:L} (\ref{line:del-0:data:ldel:R}) and in Line \ref{line:del-0:data:lin:L} (\ref{line:del-0:data:lin:R})
of the Algorithm~\ref{alg:concur}, 
depending on whether $\ms{curr}$ is left or right child.
The exact position is calculated as what comes last: Line \ref{line:del-0:data:ldel:L} (\ref{line:del-0:data:ldel:R}) or
the last invocation of unsuccessful $\lit{insert}(v)$ or $\lit{contains}(v)$
that reads the node $curr$.
\item $\ms{curr}$ is a leaf with a routing parent. Then $l_{\pi}$ is associated between the writes in
Line \ref{line:del-0:routing:ldel:curr:L} (\ref{line:del-0:routing:ldel:curr:R}) and in Line
\ref{line:del-0:routing:lin:L} (\ref{line:del-0:routing:lin:R}) of the Algorithm~\ref{alg:concur}, 
depending on whether $\ms{prev}$ is left or right child.
The exact position is calculated as what comes last: Line \ref{line:del-0:routing:ldel:curr:L}
(\ref{line:del-0:routing:ldel:curr:R}) or the last invocation of unsuccessful $\lit{insert}(v)$ or $\lit{contains}(v)$
that reads the node $curr$.
\end{enumerate}
For every $\pi = \lit{delete}(v)$ that returns $\false$, we have three cases:
\begin{enumerate}
\item If there exists a successful $\lit{delete}(v)$ whose linearization point lies in the interval of $\pi$,
then we take the first such $\pi' = \lit{delete}(v)$ and linearize right after $l_{\pi'}$.
\item If there exists successful $\lit{insert}(v)$ whose linearization point lies in the interval of $\pi$,
then we take the first such $\pi' = \lit{insert}(v)$ and linearize right before $l_{\pi'}$.
\item Otherwise, $l_{\pi}$ is the invocation point of $\pi$.
\end{enumerate}

\textbf{Contains.}
For $\pi = \lit{contains}(v)$ that returns $\true$, we have three cases:
\begin{enumerate}
\item If there exists successful $\lit{insert}(v)$ whose linearization point lies in the interval of $\pi$,
then we take the first such $\pi' = \lit{insert(v)}$ and linearize right after $l_{\pi'}$.
\item If there exists successful $\lit{delete}(v)$ whose linearization point lies in the interval of $\pi$,
then we take the first such $\pi' = \lit{delete}(v)$ and linearize right before $l_{\pi'}$.
\item Otherwise, $l_{\pi}$ is the invocation point of $\pi$.
\end{enumerate}
For $\pi = \lit{contains}(v)$ that returns $\false$, we have three cases:
\begin{enumerate}
\item If there exists successful $\lit{delete}(v)$ whose linearization point lies in the interval of $\pi$,
then we take the first such $\pi' = \lit{delete}(v)$ and linearize right after $l_{\pi'}$.
\item If there exists successful $\lit{insert}(v)$ which linearization point lies in the interval of $\pi$,
then we take the first such $\pi' = \lit{insert}(v)$ and linearize right before $l_{\pi'}$.
\item Otherwise, $l_{\pi}$ is the invocation point of $\pi$.
\end{enumerate}

To confirm our choice of linearization points, we need 
an auxiliary lemma.

\begin{lemma}
\label{lem:traverse:exact}
Consider the call $\pi = traverse(v)$. If BST at the moment of the invocation of $\pi$
contains the node $u$ with value $v$ 
and there is no linearization point of successful $delete(v)$ operation in the interval of $\pi$,
then $\pi$ returns $u$.
\end{lemma}
\begin{proof}
Consider a list $A(u)$ of ancestors of node $u$: $root = w_1, \ldots, w_{n - 1}, w_n = u$ (starting from the root)
in BST at the moment of the invocation of $\pi$.

Let us prove that at any point of time the child of $w_i$ in the direction of the value $v$ is $w_j$ for some $j > i$.
The only way for $w_i$ to change the proper child is to perform a physical deletion on this child.
Consider the physical deletions of $w_i$ in their order in execution.
In a base case, when no deletions happened, our invariant is satisfied. Suppose, we operated first $k$ deletions and now we consider
a deletion of $w_j$. Let $w_i$ be an ancestor of $w_j$ and $w_k$ be a child of $w_j$ in proper direction. After relinking $w_k$ becomes a child
of $w_i$ in proper direction, so the invariant is satisfied for $w_i$ because $i \leq j \leq k$, while the children of other vertices remain unchanged.

Summing up, $\pi$ starts at $root$, i.e., $w_1$, and traverse only the vertices from $A(u)$ in strictly increasing order. Thus $\pi$ eventually reaches $u$ and returns it.
\end{proof}

\begin{theorem}[Linearizability]
\label{th:lin}
The algorithm is linearizable with respect to the set type.
\end{theorem}
\begin{proof}
First, we prove the linearizability of the subhistory with only successful $\lit{insert}$ and $\lit{delete}$ operations
because other operations do not affect the structure of the tree.
Then we prove the linearizability of the subhistory with only update operations, 
i.e., successful and unsuccessful $\lit{insert}(v)$ and $\lit{delete}(v)$.
And finally, we present the proof for the history with all types of operations.

\textbf{Successful update functions.}
Let $\tilde{S}_{succ}^k$ be the prefix of $\tilde{S}$ consisting of the first $k$ 
complete successful operations $\lit{insert}(v)$ or $\lit{delete}(v)$ with respect to their linearization points.
We prove by induction on $k$ that the sequence $\tilde{S}_{succ}^k$ is consistent with respect to the \emph{set} type.

The base case $k = 0$, i.e., there are no complete operations, is trivial.

The transition from $k$ to $k + 1$. Suppose that $\tilde{S}_{succ}^k$ is consistent with the \textit{set} type.
Let $\pi$ with argument $v \in \mathbb{Z}$ and its response $r_{\pi}$ be the last operation in $\tilde{S}_{succ}^{k + 1}$.
We want to prove that $\tilde{S}_{succ}^{k + 1}$ is consistent with $\pi$. For that, we check all possible types of $\pi$.

\begin{enumerate}
\item $\pi = \lit{insert}(v)$ returns $\true$.

  By induction, it is enough to prove that there are no preceding operation with an argument $v$ 
  or the last preceding operation with an argument $v$ in $\tilde{S}_{succ}^{k + 1}$ is $\lit{delete}(v)$.
  Suppose the contrary: let the last preceding operation with an argument $v$ be $\pi' = \lit{insert}(v)$.
  We need to investigate two cases of insertion: whether $\pi$ finds the node with value $v$ in the tree or not.

  In the first case, $\pi$ finds a node $u$ with value $v$.
  $\pi'$ should have inserted or modified $u$. Otherwise, the BST at $l_{\pi}$
  would contain two vertices with value $v$ and this fact violates Property 1. If $\pi'$ has inserted $u$, then
  $\pi$ has no choice but only to read the state of $u$ as data, which is impossible because $\pi$ is successful.
  If $\pi'$ has changed the state of $u$ to data, then $\pi$ has to read the state of $u$ as data,
  because the linearization points of $\pi'$ and $\pi$
  are guarded by the lock on state. This contradicts the fact that $\pi$ is successful.
  
  In the second case, $\pi$ does not find a node with value $v$. 
  We know that $\pi$ and $\pi'$ are both successful. Suppose for a moment that $\pi$ wants to insert $v$ as a
  child of node $p$, while $\pi'$ inserts $v$ in some other place. Then
  the tree at $l_{\pi}$ has two vertices with value $v$, violating Property 1.
  This means, that $\pi$ and $\pi'$ both want to insert $v$ as a child of node $p$.
  Because $l_{\pi'}$ precedes $l_{\pi}$ and these linearization points are guarded by the
  lock on the corresponding link of $p$, $\pi'$ takes a lock first, modifies the link to a child of $p$ and by that forces
  $\pi$ to restart. During the second traversal, $\pi$ finds newly inserted node with value $v$ by Lemma \ref{lem:traverse:exact}
  and becomes unsuccessful. The latter contradicts the fact that $\pi$ is successful.

\item $\pi = \lit{delete}(v)$ returns $\true$.

  By induction it is enough to prove that the preceding operation with 
  an argument $v$ in $\tilde{S}_{succ}^{k + 1}$ is $\lit{insert}(v)$.
  Suppose the opposite: let the last preceding operation with $v$ be 
  $\lit{delete}(v)$ or there is no preceding operation with an argument $v$. 
  If there is no such operation, then $\pi$ could not
  find a node with value $v$, otherwise, another operation should have inserted this node and
  consequently its linearization point would have been earlier.
  Thus in this case, $\pi$ cannot successfully delete, which contradicts the result of $\pi$.

  The only remaining possibility is that the previous successful operation is $\pi' = \lit{delete}(v)$. 
  Because $\pi$ is successful, it finds a non-deleted node $u$ with value $v$. $\pi'$ should have
  find the same node $u$ by Lemma \ref{lem:traverse:exact}, otherwise, the BST right before $l_{\pi'}$
  would contain two vertices with value $v$, violating Property 1. So, both $\pi$ and $\pi'$
  take locks on the state of $u$ to perform an operation.
  Because $l_{\pi'}$ precedes $l_{\pi}$, $\pi'$ has taken the lock earlier and set the state of $u$ to routing or
  marks $u$ as deleted.
  When $\pi$ obtains the lock, it could not
  read state as data and, as a result, cannot delete the node. This contradicts the fact that $\pi$ is successful. 
\end{enumerate}

\textbf{Update operations.}
Let $\tilde{S}_m^k$ be the prefix of $\tilde{S}$ consisting of the first $k$ 
complete operations $\lit{insert}(v)$ or $\lit{delete}(v)$ 
with respect to their linearization points.
We prove by induction on $k$ that the sequence $\tilde{S}_m^k$ is consistent with respect to the \textit{set} type. 
We already proved that successful operations are consistent, then
we should prove that the linearization points of unsuccessful operations are consistent too.

The base case $k = 0$, i.e., there are no complete operations, is trivial.

The transition from $k$ to $k + 1$. Suppose that $\tilde{S}_{m}^k$ is consistent with the \textit{set} type.
Let $\pi$ with argument $v \in \mathbb{Z}$ and response $r_{\pi}$ be the last operation in $\tilde{S}_{m}^{k + 1}$.
We want to prove that $\tilde{S}_{m}^{k + 1}$ is consistent with $\pi$. For that, we check all the possible types of $\pi$.

If $k + 1$-th operation is successful then it is consistent with the previous operations, because it is consistent with successful operations
while unsuccessful operations do not change the structure of the tree.

If $k + 1$-th operation is unsuccessful, we have two cases.
\begin{enumerate}
\item $\pi = \lit{insert}(v)$ returns $\false$.
  When we set the linearization point of $\pi$ relying on the successful operation in the interval
  of $\pi$, the linearization point is correct: if we linearize right after successful $\lit{insert}(v)$
  then $\pi$ correctly returns $\false$; if we linearize right before successful $\pi' = \lit{delete}(v)$
  then by the proof of linearizability for successful operations there exists successful $\lit{insert}(v)$
  preceding $\pi'$, thus $\pi$ correctly returns $\false$.

  It remains to consider the case when no successful operation was linearized in the interval of $\pi$.
  By induction, it is enough to prove that the last preceding successful operation 
  with $v$ in $\tilde{S}_{m}^{k + 1}$ is $\lit{insert}(v)$.
  Suppose the opposite: let the last preceding successful operation with an argument $v$ 
  be $\lit{delete}(v)$ or there is no preceding operation with an argument $v$.
  If there is no such operation then $\pi$ could not find a node with value $v$,
  because, otherwise, another operation should have inserted the node and its linearization point would have come earlier. 
  Thus $\pi$ can successfully insert a new node with value $v$, which contradicts the fact that $\pi$ is unsuccessful.

  The only remaining possibility is that the last preceding successful operation is $\pi' = \lit{delete}(v)$.
  Since $l_{\pi'}$ does not lie inside the interval of $\pi$ then $\pi$ has to find either the routing node
  with value $v$ or do not find such node, since $\pi'$ has unlinked it. In both cases, insert operation could be
  performed successfully. 
  This contradicts the fact that $\pi$ is unsuccessful.

\item $\pi = \lit{delete}(v)$ returns $\false$.

  When we set the linearization point of $\pi$ relying on the successful operation in the interval of $\pi$, 
  the linearization point is correct: if we linearize right after successful 
  $\lit{delete}(v)$ then $\pi$ correctly returns $\false$;
  if we linearize right before successful $\pi' = \lit{insert}(v)$ 
  then by the proof of linearizability for successful operations 
  there exists successful $\lit{delete}(v)$ preceding $\pi'$ or there are no successful operation 
  with an argument $v$ in $\tilde{S}_{m}^{k + 1}$ before $\pi'$, thus $\pi$ correctly returns $\false$.

  It remains to consider the case when no successful operation was linearized in the interval of $\pi$.
  By induction, it is enough to prove that there is no preceding successful operation with $v$ or 
  the last preceding successful operation with $v$ in $\tilde{S}_{m}^{k + 1}$ is $\lit{delete}(v)$.
  Again, suppose the opposite: let the previous successful operation with $v$ be $\pi' = \lit{insert}(v)$.

  By Lemma \ref{lem:traverse:exact} $\pi$ finds the data node $u$ with value $v$ and
  $\pi$ can successfully remove it because no other operation with argument $v$ has a linearization point during the
  execution of $\pi$. This contradicts the fact that $\pi$ is unsuccessful.
\end{enumerate}

\textbf{All operations.}
Finally, we prove the correctness of the linearization points of all operations.

Let $\tilde{S}^k$ be the prefix of $\tilde{S}$ consisting of the first $k$ complete operations 
ordered by their linearization points.
We prove by induction on $k$ that the sequence $\tilde{S}^k$ is consistent with respect to the \emph{set} type. 
We already proved that update operations are consistent, then
we should prove that the linearization points of contains operations are consistent too.

The base case $k = 0$, i.e., there are no complete operations, is trivial.

The transition from $k$ to $k + 1$. Suppose that $\tilde{S}^k$ is consistent with the \emph{set} type.
Let $\pi$ with argument $v \in \mathbb{Z}$ and its response $r_{\pi}$ be the last operation in $\tilde{S}^{k + 1}$.
We want to proof, that $\tilde{S}^{k + 1}$ is consistent for the operation $\pi$. For that, we check all the possible types of $\pi$.

If $k + 1$-th operation is $\lit{insert}(v)$ and $\lit{delete}(v)$ 
then it is consistent with the previous $\lit{insert}(v)$ and $\lit{delete}(v)$ operations
while $\lit{contains}(v)$ operations do not change the structure of the tree.

If the operation is $\pi = \lit{contains}(v)$, we have two cases:
\begin{enumerate}
  \item $\pi$ returns $\true$.

  When we set the linearization point of $\pi$ relying on a successful update operation in the interval of $\pi$,
  then the linearization point is correct:
  \begin{itemize}
    \item if we linearize right after successful $\lit{insert}(v)$, then $\pi$ correctly returns $\true$.
    \item if we linearize right before successful $\pi' = \lit{delete}(v)$, then, by the proof of the linearizability
    on successful operations, there exists successful $\lit{insert}(v)$ preceding $\pi'$, thus $\pi$ correctly returns $\true$.
  \end{itemize}

  We are left with the case when no successful operation has its linearization point in the interval of $\pi$.
  By induction, it is enough to prove that the last preceding successful operation 
  with $v$ in $\tilde{S}^{k + 1}$ is $\lit{insert}(v)$.
  Suppose the opposite: the last preceding successful operation with an argument $v$ 
  is $\lit{delete}(v)$ or there is no preceding successful operation with $v$.
  If there is no successful operation then $\pi$ could not find a node with value $v$, otherwise,
  some operation has inserted a node before and its linearization point would have come earlier.
  This contradicts the fact that $\pi$ is successful.

  It remains to check if there exists a preceding $\pi' = \lit{delete}(v)$ operation. 
  Since $l_{\pi'}$ does not lie inside the interval of $\pi$ then $\pi$ has to find either
  the routing node with value $v$ or do not find such node, since $\pi'$ has unlinked it.
  This contradicts the fact that $\pi$ returns $\true$.

  \item $\pi$ returns $\false$.

  When we set the linearization point of $\pi$ relying on a successful update operation in the interval of $\pi$,
  then the linearization point is correct: 
  \begin{itemize}
    \item if we linearize right after successful $\lit{delete}(v)$, then $\pi$ correctly returns $\false$;
    \item if we linearize right before successful $\pi' = \lit{insert}(v)$ then, 
    by the proof of linearizability on successful operations
    either there exists a preceding $\pi'$ successful $\lit{delete}(v)$ or 
    there exists no operation with an argument $v$ in $\tilde{S}$ before $\pi'$. 
    Thus $\pi$ correctly returns $\false$.
  \end{itemize}
  We are left with the case when no successful operation has its linearization point in the interval of $\pi$.
  By induction, it is enough to prove that there is no preceding successful operation with an argument $v$ 
  or the last preceding successful operation with an argument $v$ in $\tilde{S}^{k + 1}$ is $\lit{delete}(v)$.
  Again, suppose the opposite: the last preceding successful operation 
  with an argument $v$ is $\pi' = \lit{insert}(v)$.

  By Lemma\ref{lem:traverse:exact} $\pi$ finds the data node $u$ with value $v$. This contradicts
  the fact and $\pi$ should return $\false$.
\end{enumerate}                                                      
\end{proof}

\subsection{Deadlock-freedom}
\begin{theorem}[Deadlock-freedom]
The algorithm is deadlock-free: assuming that no thread fails in the
middle of its update operation, at least one live thread makes
progress by completing infinitely many operations.   
\end{theorem}
\begin{proof}
  A thread executing $\pi=\lit{contains}(v)$ makes progress in a
  finite number of its own
  steps, because $\lit{contains}$ is wait-free.
  Otherwise, take the highest ``conflicting'' node. Note if some thread $t_1$ failed to acquire a lock on this node it happens for two reasons:
  \begin{enumerate}
    \item There is another thread $t_2$ which holds a lock on this node. 
    Since we acquire locks from children to parents and since this is
    the highest conflicting node, $t_2$ successfully acquires locks on higher nodes and
    makes progress.
    \item Some locking conditions are violated: it means that between
      the traversal phase and the attempt to acquire a lock some another thread
    $t_2$ changes expected conditions. Thus, thread $t_2$ has already made progress.
  \end{enumerate}
\end{proof}

 \section{Proof of concurrency optimality}
\begin{theorem}[Optimality]
Our binary search tree implementation is concurrency-optimal with respect to the sequential algorithm
provided in Algorithm~\ref{alg:seq}.
\end{theorem}
\begin{proof}
Consider all the executions of our algorithm in which all critical sections are executed sequentially.
Since all critical sections in our algorithm contains only one operation from the sequential algorithm, the implementation
accepts all the schedules in which the operation is not restarted by failing some condition in the critical sections.
So, it is enough to show that each condition that forces the restart is crucial, i.e., if the operation
ignores it the schedule will be not observably correct schedule.


For the next discussion we have to define two values $I(T, v)$ and $D(T, v)$~---
the number of insert and delete operations with argument $v$ that made at least one write in the prefix with length $T$ of
schedule $\sigma$, later referred as $\sigma(T)$.
Since we consider the linearization of operations that performed write, $I(T, v)$ and $D(T, v)$ are exactly the number of 
successful insert and delete operations in any completion of $\sigma(T)$.
From hereon, when we talk about the completions we mean only operations that performed write.

To slightly simplify the further proof by exhaustion we look at three common situations (later referred to as Case 1, 2 or 3)
that appear under consideration, and show that they lead to not observably correct schedule:
\begin{enumerate}
\item The modification in the critical section of operation $\pi = \lit{insert}(v)$ (the case of $\lit{delete}(v)$ is considered similarly)
does not change the set of values represented by our tree,
i.e., fields left, right and state for any node reachable from the root does not change or some routing vertex becomes unlinked.
Let this modification be the $T$-th event of the current schedule $\sigma$. Consider two prefixes of this schedule:
$\sigma(T - 1)$ and $\sigma(T)$. There could happen two cases:
\begin{itemize}
\item If the value $v$ is present in the set after $\sigma(T - 1)$,
then $I(T - 1, v) = D(T - 1, v) + 1$, since $\sigma(T - 1)$ is linearizable. 
We know that $\pi$ is successful, then $I(T, v) = D(T, v) + 2$. By that, any completion of $\sigma(T)$ cannot be linearizable, meaning that
$\sigma$ is not observably correct.
\item If the value $v$ is not present in the set after completion of $\sigma(T - 1)$ then
$I(T - 1, v) = D(T - 1, v)$, since $\sigma(T - 1)$ is linearizable.
We know that $\pi$ is successful, then $I(T, v) = D(T, v) + 1$, but the value $v$ is still not present in the set after $\sigma(T)$.
By that, any completion $\sigma(T)$ cannot be linearizable, meaning that $\sigma$ is not observably correct.
\end{itemize}
\item After the modification in the critical section of operation $\pi$ with argument $v$
a whole subtree of node $u$ with a value different from $v$ becomes unreachable from the root.
Let this modification be the $T$-th event of the current schedule $\sigma$.
Because of the structure of the tree, subtree of node $u$ should contain at least one data vertex with value $x$ not equal to $v$.
Since $x$ was reachable after the modification and $\sigma(T - 1)$ is linearizable,
we assume $I(T - 1, x) = D(T - 1, x) + 1$. The number of successful update operations with argument $x$ does not change after 
the modification, so $I(T, x) = D(T, x) + 1$. But the value $x$ is not reachable from the root after $\sigma(T)$, meaning that
any completion of $\sigma(T)$ cannot be linearized. Thus, $\sigma$ is not observably correct.
\item After the modification in the critical section of operation $\pi$ the node $u$ with deleted mark becomes reachable from the root.
Let this modification be the $T$-th event of the current schedule $\sigma$. 
Let the modification that was done in the same critical section as the deleted mark of $u$ was set to be the $\tilde{T}$-th event of $\sigma$.
It could be seen that $u$ is reachable from the root after $\sigma(\tilde{T} - 1)$ and
after $\sigma(T)$, but $u$ is not reachable from the root after
$\sigma(\tilde{T})$. Thus, $\sigma(T)$ does not satisfy the third requirement to observably correct schedule,
meaning that $\sigma$ is not observably correct.
\end{enumerate}

Now, we want to prove that all conditions that precede each modification operation are necessary and their omission
leads to not observably correct schedule. The proof is done by induction on the position of modification operation in the execution.
The base case, when there are no modification operations done, is trivial. Suppose, we show the correctness of our statement
for the first $i - 1$ modifications and want to prove it for the $i$-th.
Let this modification be the $T$-th event of the schedule $\sigma$. We ignore each condition that
precedes the modification one by one in some order
and show that their omission makes $\sigma$ not observably correct:
\begin{itemize}
\item Operation $\pi=\lit{insert}(v)$ restarts in Line \ref{line:ins:curr} of Algorithm~\ref{alg:concur}.
This means, that at least one of the following condition holds:
\begin{itemize}
\item $curr$ is not a routing node (Line \ref{line:ins:curr}). Then the guarded operation does not change the set of values and by Case 1
$\sigma$ is not observably correct.
\item Deleted mark of $curr$ is set (later, we simply say $curr$ is deleted) (Line \ref{line:ins:curr}).
then $curr$ is already unlinked, so the modification in Line \ref{line:ins:lin}
does not change the set of values and by Case 1 $\sigma$ is not observably correct.
\end{itemize}

\item Operation $\pi=\lit{insert}(v)$ restarts in Lines \ref{line:ins-null:prev:null:L}, \ref{line:ins-null:prev:deleted:L}
(\ref{line:ins-null:prev:null:R}, \ref{line:ins-null:prev:null:L}) of Algorithm~\ref{alg:concur}.
This means, that at least one of the following conditions holds:
\begin{itemize}
\item $prev$ is deleted (Line \ref{line:ins-null:prev:null:L} (\ref{line:ins-null:prev:null:R})).
Then $prev$ is already unlinked and is not reachable from the root.
This means, that the modification in Line \ref{line:ins-null:lin:L} (\ref{line:ins-null:lin:R})
links the new vertex to already unlinked vertex $prev$, not changing the set of values, and
by Case 1 $\sigma$ is not observably correct.
\item The corresponding child of $prev$ is not null (Line \ref{line:ins-null:prev:null:L} (\ref{line:ins-null:prev:null:R})).
Then the write in Line \ref{line:ins-null:lin:L} (\ref{line:ins-null:lin:R})
unlinks a whole subtree of the current child and by Case 2 $\sigma$ is not observably correct.
\end{itemize}

\item Operation $\pi=\lit{delete}(v)$ restarts in Lines \ref{line:del-2:curr:data}, \ref{line:del-2:curr:2children} of Algorithm~\ref{alg:concur}.
This means that either $curr$ is not a data node or $curr$ does not have two children.
\begin{itemize}
\item If $curr$ is not a data node or it is deleted (Line \ref{line:del-2:curr:data}),
then the write at Line \ref{line:del-2:lin} does not change the set of values and
by Case 1 $\sigma$ is not observably correct.
\item If $curr$ does not have two children (Line \ref{line:del-2:curr:2children}),
then after the write in Line \ref{line:del-2:lin} the tree has the routing node $curr$
with less than two children. Thus, after $\sigma(T)$ the tree does not satisfy the second requirement to
observably correct schedule, meaning that $\sigma$ is not observably correct.
\end{itemize}

\item Operation $\pi=\lit{delete}(v)$ restarts in Lines \ref{line:del-1:curr:child}-\ref{line:del-1:curr:1child} of Algorithm~\ref{alg:concur}.
This means that at least one of the following conditions holds:
\begin{itemize}
\item $prev$ is deleted (Line \ref{line:del-1:prev:curr}).
Then the write at Line \ref{line:del-1:lin:L} (\ref{line:del-1:lin:R}) does not change the set of values
and by Case 1 $\sigma$ is not observably correct. For later cases, we already assume, that
$prev$ is not deleted.
\item $child$ because $child$ is deleted (Line \ref{line:del-1:curr:child}).
Then after the write at Line \ref{line:del-1:lin:L} (\ref{line:del-1:lin:R}) the deleted node $curr$ becomes reachable from the root, since $prev$ is
not deleted, and by Case 3 $\sigma$ is not observably correct. From hereon, we assume that $child$ is not
deleted.
\item There is no link from $curr$ to $child$ (Line \ref{line:del-1:curr:child}).
Since $child$ is not deleted, this case could happen only if $curr$ is deleted.
We know that $prev$ and $child$ are not deleted, thus $prev$ has $child$ as its child.
By that, the write at Line \ref{line:del-1:lin:L} (\ref{line:del-1:lin:R})
does not change the set of values and by Case 1 $\sigma$ is not observably correct.
From hereon, we assume that $curr$ is not deleted.
\item There is no link from $prev$ to $curr$ (Line \ref{line:del-1:prev:curr}), because $curr$ is deleted was already covered
by the previous case.
\item $curr$ is not a data node (Line \ref{line:del-1:curr:data}). Then the write in
Line \ref{line:del-1:lin:L} (\ref{line:del-1:lin:R}) does not change the set of values
and by Case 1 $\sigma$ is not observably correct.
\item $curr$ does not have exactly one child (Line \ref{line:del-1:curr:1child}). Since none of $curr$ and $child$ are deleted,
the link from $curr$ to $child$ exists in the tree,
the only possible way to violate is that $curr$ has two children.
Thus, the write in Line \ref{line:del-1:lin:L} (\ref{line:del-1:lin:R}) unlinks a whole subtree
of the other child of $curr$ and by Case 2 $\sigma$ is not observably correct.
\end{itemize}

\item Operation $\pi=\lit{delete}(v)$ restarts in Lines \ref{line:del-0:prev:curr}-\ref{line:del-0:curr:leaf}
and \ref{line:del-0:data:prev:L} (\ref{line:del-0:data:prev:R}) of Algorithm~\ref{alg:concur}.
This means that at least one of the following conditions holds:
\begin{itemize}
\item $prev$ is deleted (Line \ref{line:del-0:data:prev:L} (\ref{line:del-0:data:prev:R})).
Then, the write in Line \ref{line:del-0:data:lin:L} (\ref{line:del-0:data:lin:R}) does not change the set of values
and by Case 1 $\sigma$ is not observably correct. From hereon, we assume that $prev$ is not deleted.
\item $prev$ is not a data node (Line \ref{line:del-0:data:prev:L} (\ref{line:del-0:data:prev:R})).
Then after the write in Line \ref{line:del-0:data:lin:L} (\ref{line:del-0:data:lin:R}) the tree contains a routing node with less than two children.
Thus, after $\sigma(T)$ the tree does not satisfy our second requirement to observably
correct schedule, meaning that $\sigma$ is not observably correct.
\item The child $c$ of $prev$ in the direction of $v$ is null (Line \ref{line:del-0:prev:curr}).
Then the write in Line \ref{line:del-0:data:lin:L} (\ref{line:del-0:data:lin:R}) does not change the set of values
and by Case 1 $\sigma$ is not observably correct.
\item The child $c$ of $prev$ in the direction of $v$ has a key different from $v$ (Line \ref{line:del-0:prev:curr}). (Note that $c$ cannot
be deleted since the link from $prev$ to $c$ is in the tree now.) The write in Line
\ref{line:del-0:data:lin:L} (\ref{line:del-0:data:lin:R}) unlinks a whole subtree of $c$
and by Case 2 $\sigma$ is not observably correct.
\item The child $c$ of $prev$ in the direction of $v$ is not a leaf (Line \ref{line:del-0:curr:leaf}).
Then the write in Line \ref{line:del-0:data:lin:L} (\ref{line:del-0:data:lin:R})
removes whole subtree of $c$ with at least one another data node and by Case 2 $\sigma$ is not observably correct.
In last case we assume that $c$ is a leaf.
\item The child $c$ of $prev$ in the direction of $v$ is a routing node (Line \ref{line:del-0:curr}).
Then before the write in Line \ref{line:del-0:data:lin:L} (\ref{line:del-0:data:lin:R})
the tree contains a routing leaf $c$. Thus, after $\sigma(T - 1)$ the tree does not
satisfy our second requirement to observably correct schedule, meaning that $\sigma$ is not observably correct.
\end{itemize}

\item Operation $\pi=\lit{delete}(v)$ restarts in Lines \ref{line:del-0:prev:curr}-\ref{line:del-0:curr:leaf}
and \ref{line:del-0:routing:prev:child:L}-\ref{line:del-0:routing:prev:L}
(\ref{line:del-0:routing:prev:child:R}-\ref{line:del-0:routing:prev:R}) of Algorithm~\ref{alg:concur}.
This means that at least one of the following conditions holds:
\begin{itemize}
\item $gprev$ is deleted (Line \ref{line:del-0:routing:gprev:prev:L} (\ref{line:del-0:routing:gprev:prev:R})).
Then the write in Line \ref{line:del-0:routing:lin:L} (\ref{line:del-0:routing:lin:R}) does not change the set of
values and by Case 1 $\sigma$ is not observably correct. From hereon, we assume that $gprev$ is not deleted.
\item $prev$ is not a child of $gprev$. (Line \ref{line:del-0:routing:gprev:prev:L} (\ref{line:del-0:routing:gprev:prev:R}))
Since $gprev$ is not deleted, this case could happen
only if $prev$ is deleted.
$prev$ could be physically deleted only if it has at most one child, thus
$curr$ or $child$ has to be deleted.
If $curr$ is deleted, then the write in Line \ref{line:del-0:routing:lin:L} (\ref{line:del-0:routing:lin:R})
does not change the set of values and by Case 1 $\sigma$ is not observably
correct. Otherwise, $child$ is deleted, then the write in Line \ref{line:del-0:routing:lin:L} (\ref{line:del-0:routing:lin:R})
links the deleted node back to the tree and by Case 3 $\sigma$
is not observably correct. Later, we assume that $prev$ is not deleted.
\item $prev$ is a data node (Line \ref{line:del-0:routing:prev:L} (\ref{line:del-0:routing:prev:R})).
Then the write in Line \ref{line:del-0:routing:lin:L} (\ref{line:del-0:routing:lin:R}) unlinks $prev$ from the tree and by
the same reasoning as in Case 2 $\sigma$ is not observably correct.
\item $child$ is not a current child of $prev$ (Line \ref{line:del-0:routing:prev:child:L} (\ref{line:del-0:routing:prev:child:R})).
$child$ should be deleted, since $prev$ is not.
Then the write in Line \ref{line:del-0:routing:lin:L} (\ref{line:del-0:routing:lin:R})
links the deleted node back to the tree and by Case 3 $\sigma$ is not observably correct.
\item The last four cases are identical to the last four cases for $\lit{delete}(v)$ that restarts in
Lines \ref{line:del-0:prev:curr}-\ref{line:del-0:curr:leaf} and \ref{line:del-0:routing:prev:L} (\ref{line:del-0:routing:prev:R}).
\end{itemize}                                                                                                                           
\end{itemize}

We showed that restart of operation in the execution happens only if the corresponding sequential schedule is not observably correct.
Thus, our algorithm is indeed concurrency-optimal.
\end{proof}

\end{document}
